\documentclass[11pt]{article}
\usepackage{amsfonts}
\usepackage{makeidx}
\usepackage{amssymb}
\usepackage{graphicx}
\usepackage{amsmath}
\usepackage[T1]{fontenc}
\usepackage{amsfonts}
\usepackage[normalem]{ulem}

\newtheorem{theorem}{Theorem}

\newtheorem{corollary}[theorem]{Corollary}

\newtheorem{adefinition}[theorem]{Definition}
\newenvironment{definition}{\begin{adefinition}\rm}{\end{adefinition}}
\newtheorem{aexample}[theorem]{Example}

\newtheorem{lemma}[theorem]{Lemma}

\newtheorem{proposition}[theorem]{Proposition}
\newtheorem{aremark}[theorem]{Remark}

\numberwithin{equation}{section} \numberwithin{theorem}{section}
\newenvironment{proof}[1][Proof]{\textbf{#1.} }{\ \rule{0.5em}{0.5em}}
\DeclareMathOperator\Tr{Tr}
\makeatother

\begin{document}

\title{On a Limiting Distribution of Singular Values \\
of Random Band Matrices }
\author{A. Lytova \\
%EndAName
{\small {\textit{Department of Mathematical and Statistical Sciences,
University of Alberta}} }\\
{\small {\textit{Edmonton, Alberta, Canada, T6G 2G1}} }\\
{\small {E-mail: lytova@ualberta.ca} }\\
\\
L. Pastur \\
{\small {\textit{Theoretical Division, B. Verkin Institute for Low
Temperature Physics }} }\\
{\small {\textit{and Engineering,}} }\\
{\small {\textit{47 Lenin Ave., Kharkiv 61103, Ukraine}} }\\
{\small {E-mail: pastur@ilt.kharkov.ua}}}
\date{}
\maketitle

\begin{abstract}
An equation is obtained for the Stieltjes transform of the normalized
distribution of singular values of non-symmetric band random matrices in the
limit when the band width and rank of the matrix simultaneously tend to
infinity. Conditions under which this limit agrees with the quarter-circle
law are found. An interesting particular case of lower triangular random
matrices is also considered and certain properties of the corresponding
limiting singular value distribution are given.
\medskip

\medskip

{\it Mathematics Subject Classifcation 2010:} 60B10, 60B20 (primary); 15B52, 15A18
(secondary).
\end{abstract}

\section{Introduction: Problem and Main Results}

\label{intro} Given a positive integer $n=2m+1$, $m\in \mathbb{N}$ consider
the $n\times n$ random matrix
\begin{equation}
A_{n}=\{A_{jk}^{(n)}\}_{|j|,|k|\leq m},\quad A_{jk}^{(n)}=b_{n}^{-1/2}v\big({%
(j-k)}/{b_{n}}\big)a_{jk}^{(n)},  \label{Ajk}
\end{equation}%
where $\{a_{jk}^{(n)}\}_{|j|,|k|\leq m}$ are real random variables, $%
\{b_{n}\}$ is a sequence of positive integers such that %\begin{equation}
%\lim_{n \to \infty} b_{n}=\infty ,  \label{bn}
%\end{equation}%
%and there exists the finite or infinite limit
\begin{equation}
\lim_{n \to \infty} b_{n}=\infty, \quad \nu := \lim_{n\rightarrow \infty
}\nu _{n}\in [1,\infty], \quad \nu _{n}=n/2b_{n},  \label{nu}
\end{equation}%
$v:$ $\mathbb{R\rightarrow R}$ is a piecewise continuous function of compact
support and we denote
\begin{align}
& \int_{-\nu }^{\nu }v^{2}(t)dt=w^{2}<\infty ,  \label{w2} \\
& \max_{t\in \mathbb{R}}v^{2}(t)=K<\infty .  \label{K}
\end{align}%
In particular, if $v=\chi _{\lbrack 0,1]}$, the indicator of the interval $%
[0,1]$, then the matrix elements $\{A_{jk}^{(n)}\}_{|j|,|k|\leq m}$ are
non-vanishing only in the "band" of the width $b_{n}$ under the principal
diagonal. If in addition $2\nu =1$, then $A_{n}$ is a lower triangular
matrix asymptotically.

We are interested in the limiting distribution of the squares of singular
values of $A_{n}$, i.e., the eigenvalues
\begin{equation}
0\leq \lambda _{1}^{(n)}\leq ...\leq \lambda _{n}^{(n)}<\infty  \label{eigv}
\end{equation}%
of the positive definite random matrix%
\begin{equation}
M_{n}=A_{n}A_{n}^{T}.  \label{Mn}
\end{equation}%
To this end we introduce the Normalized Counting Measure $N_{n}$ of (\ref%
{eigv}), setting for any interval $\Delta \subset \mathbb{R}$
\begin{equation}
N_{n}(\Delta )=\mathrm{Card}\{l\in \lbrack 1,n]:\lambda _{l}^{(n)}\in \Delta
\}/n.  \label{NCM}
\end{equation}%
It is convenient to write matrix $M_{n}$ in the form
\begin{equation}
M_{n}=\sum_{|k|\leq m}\mathbf{y}_{k}\otimes \mathbf{y}_{k},  \label{Mny}
\end{equation}%
where
\begin{equation}
\mathbf{y}_{k}=(A_{-mk}^{(n)},...,A_{mk}^{(n)})^{T}  \label{yk}
\end{equation}%
are the columns of $A_{n}$ (see (\ref{Ajk})).
%  and $L_\mathbf{y}=\mathbf{y}\otimes \mathbf{y}$ is  an  $n\times n$ rank-one matrix defined as  \begin{equation}
%   L_{\mathbf{y}}\mathbf{x}=(\mathbf{y},\mathbf{x})\mathbf{y},\;\forall \mathbf{x}\in \mathbb{R}^{n},
%      \label{LYX}
%  \end{equation}%
% i.e. $L_\mathbf{y}$ is a projector on $\mathbf{y}$.
According to (\ref{Ajk}), each $\mathbf{y}_{k}$ corresponds to the vector
\begin{equation}
\mathbf{a}_{k}=(a_{-mk}^{(n)},...,a_{mk}^{(n)})^{T}.  \label{ak}
\end{equation}%
We will assume that $\{\mathbf{a}_{k}\}_{|k|\leq m}$ are jointly independent
random vectors, however the components of each vector $\mathbf{a}_{k}$ can
be dependent. Here are the corresponding definitions \cite{GLPP:14,Pa-Pa:09}.

\begin{definition}
(i). [Isotropic vectors] \label{d:isotropic} A random vector $\mathbf{a}%
=(a_{-m},...,a_{m})\in \mathbb{R}^{n}$ is called \textit{isotropic} if
\begin{equation}
\mathbf{E}\{a_{j}\}=0,\quad \mathbf{E}\{a_{j}a_{k}\}=\delta _{jk},\quad
|j|,|k|\leq m.  \label{iso}
\end{equation}

(ii). [Unconditional distribution] The distribution of random vector $%
\mathbf{a}\in \mathbb{R}^{n}$ is called \textit{unconditional} if its
components $(a_{-m},...,a_{m})$ have the same joint distribution as $(\pm
a_{-m},...,\pm a_{m})$ for any choice of signs.

(iii). [Log-concave measure] A measure $\mu $ on $\mathbb{C}^{n}$ is \textit{%
log-concave} if for any measurable subsets $A,B$ of $\mathbb{C}^{n}$ and any
$\theta \in \lbrack 0,1]$,
\begin{equation*}
\mu (\theta A+(1-\theta )B)\geq \mu (A)^{\theta }\mu (B)^{(1-\theta )}
\end{equation*}%
whenever $\theta A+(1-\theta )B=\{\theta X_{1}+(1-\theta )X_{2}\,:\,X_{1}\in
A,\;X_{2}\in B\}$ is measurable.
\end{definition}

\begin{definition}
\lbrack Good vectors] \label{d:good} We say that a random vector $\mathbf{a}%
=(a_{-m}^{(n)},...,a_{m}^{(n)})\in \mathbb{R}^{n}$ is \textit{good,} if
% \noindent \uline{Case 1.} the components
%  $\{a_{j }\}_{|j|\leq m}$ are jointly independent random variables satisfying \begin{equation}
% \mathbf{E}\{a_{j}\}=0,\quad  \mathbf{E}\{a^{2}_{j}%
% \}=1, \quad  \mathbf{E}\{|a_{j}|^{2+\varepsilon}\}<\infty,\; \varepsilon>0. \label{mom}
% \end{equation}%
it is an isotropic vector with an unconditional distribution satisfying the
moment conditions
\begin{equation}  \label{m22}
m_{2,2}^{(n)}= \mathbf{E}\{(a_{j}^{(n)})^2 (a_{k}^{(n)})^2\}=1+o(1),\;j\neq
k,\quad m_{4}^{(n)} = \mathbf{E}\{(a_{j}^{(n)})^4\}=O(1)
\end{equation}%
as $n=2m+1\rightarrow \infty $. %(we do not write here and
%below the superindex $n$ for the components of vectors).
Note that $m_{2,2}^{(n)}$ and $m_{4}^{(n)}$ do not depend on $j$ and $k$.
\end{definition}

A simple example of good vectors are the vectors with i.i.d. $n$-independent
components of zero mean and unit variance. An important case is given by the
isotropic random vectors with symmetric unconditional and log-concave
distributions (see Lemma 2.1 of \cite{GLPP:14}), the simplest among them are
the vectors uniformly distributed over the unit sphere in $\mathbb{R}^{n}$.
Now we are ready to formulate our main results.

\begin{theorem}
\label{t:main} Let $M_{n},\;n=2m+1$, $m\in \mathbb{N}$ be the random matrix (%
\ref{Mny}) -- (\ref{ak}), where for every $m$ $\{\mathbf{a}_{k}\}_{|k|\leq
m} $ are jointly independent good vectors (see Definition \ref{d:good}) and
corresponding vectors $\{\mathbf{y}_{k}\}_{|k|\leq m}$ are defined in (\ref%
{yk}) and (\ref{Ajk}) -- (\ref{K}). Let $N_{n}$ be the Normalized Counting
Measure (\ref{NCM}) of eigenvalues of $M_{n}$.. Then there exists a
non-random and non-negative measure $N$, $N(\mathbb{R})=1$ such that for any
interval $\Delta \subset \mathbb{R}$ we have in probability
\begin{equation}
\lim_{n\rightarrow \infty }N_{n}(\Delta )=N(\Delta ).  \label{NnN}
\end{equation}%
The limiting measure $N$ is uniquely defined via its Stieltjes transform $f$
(see \cite{Ak-Gl:93,Pa-Sh:11})
\begin{equation}
f(z)=\int_{0}^{\infty }\frac{N(d\lambda )}{\lambda -z},\;\Im z\neq 0,
\label{f}
\end{equation}%
by the formula
\begin{equation}
\int_{\mathbb{R}}\varphi (\lambda )N(d\lambda )=\lim_{\varepsilon
\rightarrow 0^{+}}\frac{1}{\pi }\int_{\mathbb{R}}\varphi (\lambda )\Im
f(\lambda +i\varepsilon )d\lambda ,\quad \forall \varphi \in C_{0}(\mathbb{R}%
),  \label{SP}
\end{equation}%
and we have:

\smallskip (i)\ if $\nu <\infty $ in (\ref{nu}), then
\begin{equation}
f(z)=\frac{1}{2\nu }\int_{-\nu }^{\nu }f(t,z)dt,  \label{ftz}
\end{equation}%
where $f:\lbrack -\nu ,\nu ]\times \mathbb{C}\setminus \lbrack {0},\infty
)\rightarrow \mathbb{C}$ on the right of the formula is continuous in $t$
for every $z\in \mathbb{C}\setminus \lbrack {0},\infty )$, is analytic in $z$
for every $|t|\leq \nu $, has the property%
\begin{equation}
\Im f(t,z)\Im z\geq 0,\quad |f(t,z)|\leq |\Im z|^{-1},\;\Im z\neq 0,
\label{Imzf}
\end{equation}%
and is the unique solution of the equation
\begin{equation}
f(t,z)=-\Bigg(z-\int_{-\nu }^{\nu }\frac{v^{2}(t-\tau )d\tau }{1+\int_{-\nu
}^{\nu }v^{2}(\theta -\tau )f(\theta ,z)d\theta }\Bigg)^{-1}.  \label{eq}
\end{equation}

\smallskip (ii) \ if $\nu =\infty $ in (\ref{nu}), then $f$ of (\ref{f}) is
the unique solution of the quadratic equation
\begin{equation}
zw^{2}f^{2}+zf+1=0  \label{eqc}
\end{equation}%
in the class of analytic in $\mathbb{C\setminus \mathbb{R}}$ functions
satisfying
\begin{equation}
\Im f(z)\Im z\geq 0,\quad |f(z)|\leq |\Im z|^{-1},\;\Im z\neq 0,  \label{Nev}
\end{equation}%
and we have the following formula for the density $\rho _{qc}$ of the
limiting measure $N$:
\begin{equation}
N(d\lambda )=\rho _{qc}(\lambda )d\lambda ,\quad \rho _{qc}(\lambda )=(2\pi
w^{2})^{-1}\sqrt{(4w^{2}-\lambda )/\lambda }\mathbf{1}_{[0,4w^{2}]},
\label{qc}
\end{equation}%
known as the quarter-circle law.
\end{theorem}

\begin{theorem}
\label{t:indep} The results of Theorem \ref{t:main} remain valid if $\{%
\mathbf{a}_{k}\}_{|k|\leq m}$ are independent isotropic random vectors with
independent components having finite absolute moment of the order $%
2+\varepsilon $, $\varepsilon >0$,
\begin{equation}
\sup_{n} \max_{|j|,|k|\leq m}\mathbf{E}\{|a_{kj}^{(n)}|^{2+\varepsilon
}\}<\infty .  \label{a2e}
\end{equation}
\end{theorem}

\begin{corollary}
\label{l:qc} Under conditions of Theorem \ref{t:main} (or Theorem \ref%
{t:indep}) with $\nu <\infty $ in (\ref{nu}) $N$ is the quarter-circle law
if and only if the function $v^{2}:[-2\nu ,2\nu ]\rightarrow \mathbb{R}_{+}$
is the restriction on the interval $[-2\nu ,2\nu ]$ of a $2\nu $-periodic
function.
\end{corollary}

In particular, if all entries $\{A_{jk}^{(n)}\}_{|j|,|k|\leq m}$ are
non-vanishing, then we get the quarter-circle law, and this fact was proved
long time ago \cite{Ma-Pa:67} (see also \cite{Pa-Pa:09,Pa-Sh:11}).

To prove the corollary, we note first that if $\nu <\infty $ and if $v^{2}$
is $2\nu $-periodic, then (\ref{eq}) has $t$-independent solution $f$,
satisfying (\ref{eqc}). If $v^{2}$ does not possess this property, then the
function%
\begin{equation*}
u(t)\equiv \int_{-\nu }^{\nu }v^{2}(t-\tau )d\tau
\end{equation*}%
cannot be a constant on the interval $(-\nu ,\nu )$. Hence, expanding the
solution $f$ of (\ref{eq}) in the inverse powers of $z$:
\begin{equation*}
f(z)=-\frac{1}{z}-\frac{a_{1}}{z^{2}}-\frac{a_{2}}{z^{3}}+O(z^{-4}),\quad
a_{1}=\frac{1}{2\nu }\int_{-\nu }^{\nu }u(t)dt,\quad a_{2}=\frac{1}{2\nu }%
\int_{-\nu }^{\nu }u^{2}(t)dt,
\end{equation*}%
and then applying the Schwarz inequality, we get the strict inequality
\begin{equation*}
a_{1}^{2}<a_{2}/2,
\end{equation*}%
if $u$ is not identically constant. On the other hand,\ we have for $f_{qc}$
of (\ref{qc})
\begin{equation*}
f_{qc}(z)=\frac{1}{2w^{2}}(-1+\sqrt{(1-4w^{2}/z)})=-\frac{1}{z}-\frac{w^{2}}{%
z^{2}}-\frac{2w^{4}}{z^{3}}+O(z^{-4}),\;z\rightarrow \infty ,
\end{equation*}%
so that $a_{1}^{2}=a_{2}/2$. Therefore, in the considered case of a
non-periodic $v^{2}$, the limiting Normalized Counting Measure $N$ of (\ref%
{NnN}) cannot be the quarter-circle law.

Note that the results of Theorem \ref{t:main} and ~ Corollary \ref{l:qc}
agree with those obtained in \cite{C-G:93} and \cite{MPK:92} for the Wigner
band matrices. Indeed, according to \cite{C-G:93} and \cite{MPK:92} the
Stieltjes transform of the Wigner matrices (i.e., matrices with independent,
modulo symmetry conditions, entries satisfying an analog of (\ref{a2e})) is
given by the same formula (\ref{ftz}), where now $f(t,z)$ solves uniquely
the equation (cf. (\ref{eq}))%
\begin{equation*}
f(t,z)=-\Bigg(z+\int_{-\nu }^{\nu }v^{2}(t-\tau )f(z,\tau )d\tau \Bigg)^{-1}
\end{equation*}%
in the same class of functions defined by (\ref{Imzf})).

Another corollary of Theorem \ref{t:main} yields the liming distribution of
singular values of lower triangular random matrices.

\begin{theorem}
\label{t:tri} Consider the case of Theorem \ref{t:main}, where $v=\chi
_{[0,1]}$ is the indicator of the interval $[0,1]$ and $2\nu =1$, so that
the matrices $A_{n}$ are lower triangular. Then:

\smallskip (i) \ the Stieltjes transform $f$ (\ref{f})~of the limiting
Normalized Counting Measure (\ref{NnN}) of eigenvalues of \ $M_{n}$ (\ref{Mn}%
) solves uniquely the equation
\begin{equation}
(1+f(z))\ln (1+f(z))=-z^{-1},  \label{eqtri}
\end{equation}%
in the class of functions analytic in $\mathbb{C}\setminus \lbrack 0,\infty
) $ and satisfying (\ref{Nev});

\smallskip (ii) \
\begin{equation}
\mathrm{supp}\;N=[0,e];  \label{supp}
\end{equation}

\smallskip (iii) \ the measure $N$ of (\ref{NnN}) is absolutely continuous
and its density $\rho $ has the following asymptotics at the endpoints of
its support $[0,e]$:%
\begin{eqnarray}  \label{N0}
\rho (\lambda )&=&\frac{1}{\lambda (\ln \lambda)^2 }(1+o(1)),\;\lambda
\downarrow 0, \\
\rho (\lambda )&=&\mathrm{const}\cdot (e-\lambda )^{1/2}(1+o(1)),\;\lambda
\uparrow e.\;  \notag
\end{eqnarray}

\smallskip (iv) \ moments of $N$, i.e., $\mu _{k}:=\lim_{n\rightarrow \infty
}n^{-1}\mathrm{Tr}M_{n}^{k}$ are%
\begin{equation}
\mu _{k}=\frac{k^{k}}{(k+1)!}.  \label{moms}
\end{equation}
\end{theorem}

\textbf{Remarks}. (1). \ The lower edge $\lambda =0$ of the support is a
hard edge in the random matrix terminology. The typical \ (or standard) soft
edge asymptotic of the density of the limiting Normalized Counting Measure
near the hard edge is $\rho (\lambda )=const\cdot \lambda ^{-1/2},\;\lambda
\downarrow 0^{+\text{ }}$ \cite{Pa-Sh:11}. The asymptotics (\ref{N0}) for
the lower triangular matrices seems the most singular among the known so
far. It follows from the results of \cite{Dy-Ha:04} that for the matrices $%
M_{n}^{(q)}=$ $(A_{n})^{q}((A_{n})^{q})^{T}$ the soft edge asymptotic of the
corresponding density is $\rho (\lambda )=const\cdot (\lambda (\ln
1/\lambda )^{q+1})^{-1}$.

\smallskip (2).\ It is of interest that if we replace the lower triangular
matrix $A_{n}$ by $A_{n}+yI_{n}$ where $y>0$ and $I_{n}$ is $n\times n$ unit
matrix, than it can be shown that the support of the corresponding limiting
distribution is $[a_{\_}(y),a_{+}(y)]$, for any $y>0,
\;a_{\_}(y)>0,\;a_{+}(y)<\infty$,\ but $a_{-}(y)\rightarrow
0^{+},\;a_{+}(y)\rightarrow e $ as $y\rightarrow 0$ and the both edges of
support are soft.

\smallskip (3). \ Formula (\ref{moms}) was found in \cite{Dy-Ha:04} by
combining the operator and the free probability methods. Our proof is based
on the random matrix theory.

\medskip

To conclude the section we note that the results of Theorem \ref{t:main} for
$\nu <\infty $, in particular those for the triangular random matrices
generalize in part various results of works \cite%
{CGLZ:13,Dy-Ha:04,Gi:01,HLN:07} obtained for matrices with independent
entries by various methods. In Section \ref{s:indep} we outline the proof of
Theorem \ref{t:indep} treating the case of independent entries under
condition (\ref{a2e}), applicable for both finite and infinite $\nu $ and
based on the scheme developed in \cite{Pa-Sh:11} to find the limiting
eigenvalue distribution of a wide variety of random matrices.

\section{Proof of Theorem \protect\ref{t:main}}

\label{s:proof} Recall that if $m$ is a non-negative measure of unit mass
and
\begin{equation}
s(z)=\int_{-\infty }^{\infty }\frac{m(d\lambda )}{\lambda -z},\;\Im z\neq 0,
\label{ST}
\end{equation}%
is its Stieltjes transform, then this correspondence is one-to-one, provided
that
\begin{equation}
\Im s(z)\Im z>0,\;\lim_{n\rightarrow \infty }\eta |s(i\eta )|=1.  \label{nev}
\end{equation}%
Moreover, the correspondence is continuous if we use the uniform convergence
of analytic functions on a compact set of $\mathbb{C}\setminus \mathbb{R}$
for Stieltjes transforms and the weak convergence of probability measures
(see e.g. \cite{Ak-Gl:93,Pa-Sh:11}).
%Let us note that the solution $f$ of (%
%\ref{ftz}) -- (\ref{eq}) satisfies (\ref{nev}). Hence, to prove (\ref{NnN})
%it suffices to show that $\{g_{n}\}_{n\geq 1}$ converges uniformly on any
%compact set of $\mathbb{C}\setminus \mathbb{R}$ to $f$ of (\ref{ftz}) -- (%
%\ref{eq}) in probability.%

Let
\begin{equation}
g_{n}(z)=\int_{-\infty }^{\infty }\frac{N_{n}(d\lambda )}{\lambda -z}.
\label{gn}
\end{equation}%
be the Stieltjes transform of $N_{n}$ of (\ref{NCM}).

By using the representation (\ref{Mny}) and repeating almost literally the
proof of Theorem 19.1.6 of \cite{Pa-Sh:11}, we obtain the bounds%
\begin{equation}
\mathbf{P}\{|N_{n}(\Delta )-\mathbf{E}\{N_{n}(\Delta )\}|>\varepsilon \}\leq
C(\varepsilon )/n,  \label{EN}
\end{equation}%
for any $\Delta \subset \mathbb{R}$, where $C(\varepsilon )$ is independent
of $n$ and is finite if $\varepsilon >0$. The bound, the above one-to-one
correspondence between the measures and their Stieltjes transforms and the
analyticity of $g_{n}$ of (\ref{gn}) in $\mathbb{C\setminus \mathbb{R}}$
reduce the proof of the theorem to that of the limiting relation

\begin{equation}
\lim_{n\rightarrow \infty }\mathbf{E}\{g_{n}(z)\}=f(z),\;  \label{gnjf}
\end{equation}%
uniformly on the set%
\begin{equation}
C_{K_{0}}=\{z:\Re z=0,\;|\Im z|\geq K_{0}>0\}  \label{CK0}
\end{equation}%
where $K_{0}$ is large enough (see (\ref{K1})\ and (\ref{K2})).

It follows from (\ref{NCM}) and the spectral theorem for real symmetric
matrices that%
\begin{equation}
g_{n}(z)=\frac{1}{n}\Tr G(z)=\frac{1}{n}\sum_{|j|\leq m}G_{jj}(z),
\label{Tr}
\end{equation}%
where%
\begin{equation*}
G(z)=(M_{n}-z)^{-1}
\end{equation*}%
is the resolvent of $M_{n}$. We have
\begin{equation}
|G_{jj}|\leq |\Im z|^{-1},\sum_{k}|G_{jk}|^{2}\leq |\Im z|^{-2}.  \label{G<}
\end{equation}%
Here and in what follows we use the notation%
\begin{equation*}
\sum_{j}=\sum_{|j|\leq m}.
\end{equation*}%
We have by the resolvent identity
\begin{equation}  \label{Gjj}
G_{jj}=-\frac{1}{z}+\frac{1}{z}\sum_{k}(\mathbf{y}_{k}{\otimes }\mathbf{y}%
_{k}G)_{jj}.
\end{equation}%
Let us introduce the matrix
\begin{equation}
M_{n}^{k}=\sum_{l\neq k}\mathbf{y}_{l}{\otimes }\mathbf{y}_{l},  \label{Mk}
\end{equation}%
and its resolvent

\begin{equation}
G^{k}(z)=(M_{n}^{k}-zI_{n})^{-1},\quad \Im z\neq 0.  \label{Gk}
\end{equation}%
It follows from the rank-one perturbation formula
\begin{equation}
G-G^{k}=-\frac{G^{k}\mathbf{y}_{k}{\otimes }\mathbf{y}_{k}G^{k}}{1+(G^{k}%
\mathbf{y}_{k},\mathbf{y}_{k})}  \label{G-G}
\end{equation}%
that
\begin{equation}
(\mathbf{y}_{k}{\otimes }\mathbf{y}_{k}G)_{jj}=\frac{(G^{k}\mathbf{y}%
_{k})_{j}\mathbf{y}_{kj}}{1+(G^{k}\mathbf{y}_{k},\mathbf{y}_{k})}\equiv
\frac{B_{kn}}{A_{kn}},  \label{AkBk}
\end{equation}%
and we obtain from (\ref{Gjj})
\begin{equation}
\mathbf{E}\{G_{jj}\}=-\frac{1}{z}+\frac{1}{z}\sum_{k}\mathbf{E}\Big\{\frac{%
B_{kn}}{A_{kn}}\Big\}.  \label{EGjj}
\end{equation}%
The moment conditions (\ref{iso}) and the fact that $G^{k}$ does not depend
on $\mathbf{y}_{k}$ allow us to write
\begin{equation}
\mathbf{E}_{k}\{A_{kn}\}=1+b_{n}^{-1}\sum_{p}v_{pk}^{2}G_{pp}^{k},\quad
\mathbf{E}_{k}\{B_{kn}\}=b_{n}^{-1}v_{jk}^{2}G_{jj}^{k},  \label{EABk}
\end{equation}%
where
\begin{equation*}
v_{jk}=v((j-k)/b_{n})
\end{equation*}%
and $\mathbf{E}_{k}$ denotes the expectation with respect to $\mathbf{y}_{k}$%
. It follows from (\ref{EABk}) and the identity
\begin{equation}
\frac{1}{A}=\frac{1}{\mathbf{E}\{A\}}-\frac{1}{\mathbf{E}\{A\}}\frac{%
A^{\circ }}{A},\quad A^{\circ }=A-\mathbf{E}\{A\},  \label{A}
\end{equation}%
that
\begin{equation}
\mathbf{E}\{G_{jj}\}=-\frac{1}{z}+\frac{1}{zb_{n}}\sum_{k}\frac{v_{jk}^{2}%
\mathbf{E}\{G_{jj}^{k}\}}{1+b_{n}^{-1}\sum_{p}v_{pk}^{2}\mathbf{E}%
\{G_{pp}^{k}\}}+r_{n},  \label{EGjj1}
\end{equation}%
where
\begin{equation}
r_{n}=-\frac{1}{z}\sum_{k}\frac{1}{\mathbf{E}\{A_{kn}\}}\mathbf{E}\Big\{%
\frac{(A_{kn})^{\circ }B_{kn}}{A_{kn}}\Big\}.  \label{rn}
\end{equation}%
Let us show that
\begin{equation}
r_{n}=o(1),\quad n\rightarrow \infty .  \label{rn0}
\end{equation}%
By the spectral theorem for the real symmetric matrices there exists a
non-negative measure $m^{k}$ such that
\begin{equation*}
(G^{k}\mathbf{y}_{k},\mathbf{y}_{k})=\int_{0}^{\infty }\frac{m^{k}(d\lambda )%
}{\lambda -z},
\end{equation*}%
and we can write
\begin{equation*}
\Im (z(G^{k}\mathbf{y}_{k},\mathbf{y}_{k}))=\Im \int_{0}^{\infty }\frac{%
\lambda m_{k}(d\lambda )}{\lambda -z}=\Im z\int_{0}^{\infty }\frac{\lambda
m_{k}(d\lambda )}{|\lambda -z|^{2}}.
\end{equation*}%
Thus $\Im z\;\Im (z(G^{k}\mathbf{y}_{k},\mathbf{y}_{k}))\geq 0$ and
\begin{equation}
{|A_{kn}|^{-1}}\leq \Big|\frac{z}{\Im z+\Im (z(G^{k}\mathbf{y}_{k},\mathbf{y}%
_{k}))}\Big|\leq {|z|}{|\Im z|^{-1}}=1,\quad z\in C_{K_{0}},  \label{A>}
\end{equation}%
implying the bounds
\begin{equation}
{|\mathbf{E}_{k}\{A_{kn}\}|}^{-1}\leq 1,\quad {|\mathbf{E}\{A_{kn}\}|}%
^{-1}\leq 1,\quad z\in C_{K_{0}}.  \label{EA>}
\end{equation}%
This and the Schwarz inequality allow us to write for $r_{n}$ of (\ref{rn}):
\begin{equation}
|r_{n}|\leq \frac{1}{|\Im z|}\sum_{k}\mathbf{E}\{|(A_{kn})^{\circ
}|^{2}\}^{1/2}\mathbf{E}\{|B_{kn}|^{2}\}^{1/2}.  \label{rn<}
\end{equation}%
It follows then from (\ref{m22}), (\ref{G<}) and the bounds (see (\ref{K}))
\begin{equation}
\max_{s,k}v_{sk}^{2}\leq K,\quad \max_{k}\frac{1}{b_{n}}\sum_{s}v_{sk}^{2}\leq
w^{2},  \label{vK}
\end{equation}%
valid for sufficient large $n$, that
\begin{equation}
\mathbf{E}\{|B_{k}|^{2}\}=\frac{v_{jk}^{2}}{b_{n}^{2}}\sum_{s,t}\mathbf{E}%
\{G_{js}^k\overline{G_{jt}^k} v_{sk}v_{tk}\mathbf{E}_{k}%
\{a_{sk}^{(n)}a_{tk}^{(n)}a_{jk}^{(n)2}\}\}\leq \frac{Cv_{jk}^{2}}{|\Im
z|^{2}b_{n}^2},  \label{Bk<}
\end{equation}%
where $C$ is an absolute constant. Now (\ref{rn0}) follows from (\ref{rn<})
-- (\ref{Bk<}) and (\ref{Gyy}).

Let us show that we can replace $G^{k}$ by $G$ in (\ref{EGjj1}) with
the error of the order $O(b_{n}^{-1})$. Indeed, we have from (\ref{G-G})
\begin{equation}
G_{ps}-G_{ps}^{k}=-\frac{(G^{k}\mathbf{y}_{k})_{p}(G^{k}\mathbf{y}_{k})_{s}}{%
A_{kn}}.\quad   \label{Gps}
\end{equation}%
Applying (\ref{Gps}), (\ref{EA>}) and then (\ref{iso}), (\ref{G<}) and (\ref%
{vK}), we get
\begin{align}
\mathbf{E}_{k}\{|G_{pp}-G_{pp}^{k}|\}& \leq \mathbf{E}_{k}\{|(G^{k}\mathbf{y}%
_{k})_{p}|^{2}\}=\frac{1}{b_{n}}\sum_{s,t}G_{ps}^{k}\overline{G_{pt}^{k}}%
v_{sk}v_{tk}\mathbf{E}_{k}\{a_{sk}^{(n)}a_{tk}^{(n)}\}  \notag \\
& =\frac{1}{b_{n}}\sum_{s}|G_{ps}^{k}|^{2}v_{sk}^{2}\leq \frac{K}{b_{n}|\Im
z|^{2}}.  \label{GGk}
\end{align}%
Now it follows from (\ref{EGjj1}), in which $G_{jj}^{k}$ is replaced by $%
G_{jj}$, (\ref{rn0}) and (\ref{GGk}) that
\begin{equation*}
\mathbf{E}\{G_{jj}\}=-\frac{1}{z}+\frac{1}{zb_{n}}\sum_{k}\frac{v_{jk}^{2}%
\mathbf{E}\{G_{jj}\}}{1+b_{n}^{-1}\sum_{p}v_{pk}^{2}\mathbf{E}\{G_{pp}\}}%
+r_{nj},\quad n\rightarrow \infty ,
\end{equation*}%
where we denote by $r_{nj}$ any reminder satisfying
\begin{equation*}
\sup_{z\in C_{K_{0}}}\max_{j}|r_{nj}|\rightarrow 0,\quad n\rightarrow \infty
.
\end{equation*}%
Hence, we have
\begin{equation}
\mathbf{E}\{G_{jj}\}=\Big(\frac{1}{b_{n}}\sum_{k}\frac{v_{jk}^{2}}{%
1+b_{n}^{-1}\sum_{p}v_{pk}^{2}\mathbf{E}\{G_{pp}\}}-z\Big)^{-1}(1+r_{nj}).
\label{EGjj=}
\end{equation}%
Using (\ref{A>}) ant the fact that $\Im G_{pp}(z)\Im z\geq 0$, $\Im z\neq 0,$
it is easy to show that the denominators in (\ref{EGjj=}) do not vanish.

Fix $z$ and $n$ and introduce the piece-wise constant function
\begin{equation}
f_{n}(t,z)=\left\{
\begin{array}{cc}
0, & t\notin \lbrack (-m-1)/b_{n};m/b_{n}]\  \\
\mathbf{E}\{G_{jj}(z)\}, & \quad \quad t\in ((j-1)/b_{n};j/b_{n}],\quad
|j|\leq m.%
\end{array}%
\right.  \label{fntz}
\end{equation}%
%
%
%
%
%and corresponding step-functions $v^2_n$ which approximate $v^2$.
We have from (\ref{Tr})
\begin{align}
\mathbf{E}\{g_{n}(z)\}& =\frac{b_{n}}{n}\frac{1}{b_{n}}%
\sum_{j}f_{n}(j/b_{n},z)=\frac{b_{n}}{n}%
\int_{(-m-1)/b_{n}}^{m/b_{n}}f_{n}(t,z)dt  \notag \\
& =\frac{1}{2\nu _{n}}\int_{-\nu _{n}}^{\nu _{n}}{}f_{n}(t,z)dt+o(1),\quad
n\rightarrow \infty ,  \label{gnhn}
\end{align}%
where $\nu _{n}={b_{n}}/{m}\rightarrow \nu ,\quad n\rightarrow \infty$
(see (\ref{nu})). Besides, (\ref{EGjj=}) implies
\begin{equation*}
f_{n}(j/b_{n},z)=\Big(\frac{1}{b_{n}}\sum_{k}\frac{v^{2}((j-k)/b_{n})}{%
1+b_{n}^{-1}\sum_{p}v^{2}((p-k)/b_{n})f_{n}(p/b_{n},z)}-z\Big)^{-1}+r_{nj}.
\end{equation*}%
and taking into account (\ref{w2}), we get for any $|t|\leq \nu _{n}$
\begin{equation}
f_{n}(t,z)=\Bigg(\int_{|\tau |\leq \nu _{n}}\frac{v^{2}(t-\tau )d\tau }{%
1+\int_{|\theta |\leq \nu _{n}}v^{2}(\theta -\tau )f_{n}(\theta ,z)d\theta }%
-z\Bigg)^{-1}+r_{n}(t,z),  \label{eqhn}
\end{equation}%
where
\begin{equation}
\lim_{n\rightarrow \infty }\sup_{z\in \leq C_{K_{0}}}\sup_{|t|\leq \nu
_{n}}|r_{n}(t,z)|=0.  \label{rnt}
\end{equation}%
%
%Let $f(t,z)$ be the solution of (\ref{eq}) -- (\ref{Imzf}).
Note now that (\ref{eq}) can be written as $f=Tf$ where $T$ is a contracting
map for any $z\in C_{K_{0}}$. Indeed, we have for any pair $f_{1}$, $f_{2}$
satisfying (\ref{Imzf}) and any $z\in C_{K_{0}}$
\begin{align}
& \Big|\int_{|\tau |\leq \nu }{v^{2}(t-\tau )\Big(1+\int_{|\theta |\leq \nu
}v^{2}(\theta -\tau )f_{1,2}(\theta ,z)d\theta \Big)^{-1}d\tau }-z\Big|%
^{-1}\leq K_{0}^{-1},  \label{den1} \\
& \Big|1+\int_{|\theta |\leq \nu }v^{2}(\theta -\tau )f_{1,2}(\theta
,z)d\theta \Big|^{-1}\leq (1-w^{2}/|\Im z|)^{-1}\leq (1-w^{2}/K_{0})^{-1},
\label{den2}
\end{align}%
so that
\begin{equation}
\sup_{|t|\leq \nu }|[Tf_{1}](t,z)-[Tf_{2}](t,z)|\leq q\sup_{|t|\leq \nu
}|f_{1}(t,z)-f_{2}(t,z)|,  \label{Tq}
\end{equation}%
where
\begin{equation}
q\leq \frac{w^{4}}{(K_{0}-w^{2})^{2}}<1\quad \text{if}\quad |\Im z|\geq
K_{0}>2w^{2}.  \label{K1}
\end{equation}%
Hence, for all $z\in C_{K_{0}}$ there exists a unique solution of (\ref{eq})
satisfying (\ref{Imzf}).

Consider first the case $\nu <\infty $. Then it follows from (\ref{K}) that
for any uniformly bounded in $t$, $z$, $n$ functions $\{F_{n}\}$ we have
\begin{equation}
\lim_{n\rightarrow \infty }\sup_{z\in C_{K_{0}}}\sup_{|t|\leq \nu _{n}}\Big|%
\Big\{\int_{|\tau |\leq \nu _{n}}-\int_{|\tau |\leq \nu }\Big\}v^{2}(t-\tau
)F_{n}(\tau ,z)d\tau \Big|=0.  \label{int-int}
\end{equation}%
This, (\ref{eq}), (\ref{eqhn}) -- (\ref{rnt}), and (\ref{Tq}) lead to
\begin{equation}
\lim_{n\rightarrow \infty }\sup_{z\in C_{K_{0}}}\sup_{|t|\leq \nu
_{n}}|f_{n}(t,z)-f(t,z)|=0.  \label{hnh}
\end{equation}%
hence,
\begin{equation*}
f(z)=\lim_{n\rightarrow \infty }\mathbf{E}\{g_{n}(z)\}=\lim_{n\rightarrow
\infty }\frac{1}{2\nu _{n}}\int_{|t|\leq \nu _{n}n}f(t,z)dt=\frac{1}{2\nu }%
\int_{|t|\leq \nu }f(t,z)dt.
\end{equation*}%
Consider now the case $\nu =\infty $. In this case the unique solution of (%
\ref{eqc}) is $t$-independent function $f$, satisfying (\ref{qc}). In
addition, we have by (\ref{den1})
\begin{equation}
\frac{1}{2\nu _{n}}\int_{|t|\leq \nu _{n}}|f_{n}(t,z)-f(z)|dt\leq K_{0}^{-2}%
\frac{1}{2\nu _{n}}\int_{|t|\leq \nu _{n}}\big(%
T_{1}^{(n)}+T_{2}^{(n)}+T_{3}^{(n)}\big)dt,  \label{T123}
\end{equation}%
where
\begin{align*}
& T_{1}^{(n)}=\int_{|\tau |\leq \nu _{n}}v^{2}(t-\tau )\Bigg|\Big(%
1+\int_{|\theta |\leq \nu _{n}}v^{2}(\theta -\tau )f_{n}(\theta ,z)d\theta %
\Big)^{-1} \\
& \hspace{5cm}-\Big(1+\int_{|\theta |\leq \nu _{n}}v^{2}(\theta -\tau
)f(z)d\theta \Big)^{-1}\Bigg|d\tau , \\
& T_{2}^{(n)}=\int_{|\tau |\leq \nu _{n}}v^{2}(t-\tau )\Bigg|\Big(%
1+\int_{|\theta |\leq \nu _{n}}v^{2}(\theta -\tau )f(z)d\theta \Big)^{-1} \\
& \hspace{5cm}-\Big(1+\int_{|\theta |\leq \infty }v^{2}(\theta -\tau
)f(z)d\theta \Big)^{-1}\Bigg|d\tau , \\
& T_{3}^{(n)}=|1+w^{2}h(z)|^{-1}\int_{|\tau |\geq \nu _{n}}v^{2}(t-\tau
)d\tau .
\end{align*}%
It follows from (\ref{eqc}) that
\begin{equation*}
|1+w^{2}f(z)|^{-1}=|zf(z)|\leq 1,\quad \forall z\in C_{K_{0}},
\end{equation*}%
hence,
\begin{eqnarray*}
\frac{1}{2\nu _{n}}\int_{|t|\leq \nu _{n}}T_{3}^{(n)}dt &\leq &\frac{1}{2\nu
_{n}}\int_{|t|\leq \nu _{n}}dt\int_{|\tau |\geq \nu _{n}}v^{2}(t-\tau )d\tau
\\
&=&\int_{|y|\geq \nu _{n}}v^{2}(y)dy-\frac{1}{2\nu _{n}}\int_{|y|\leq \nu _{n}}|y|v^{2}(y)dy,
\end{eqnarray*}%
and then (\ref{nu}) and (\ref{w2}) imply
\begin{equation}
\lim_{n\rightarrow \infty }\frac{1}{2\nu _{n}}\int_{|t|\leq \nu
_{n}}T_{3}^{(n)}dt=0.  \label{T3}
\end{equation}%
Furthermore, it follows from (\ref{den2}) and (\ref{K1}) that
\begin{align*}
\frac{1}{2\nu _{n}}\int_{|t|\leq \nu _{n}}T_{2}^{(n)}dt& \leq \frac{2}{\nu
_{n}}\int_{|t|\leq \nu _{n}}dt\int_{|\tau |\leq \nu _{n}}v^{2}(t-\tau
)\int_{|\theta |\geq \nu _{n}}v^{2}(\theta -\tau )d\theta \\
& =\frac{2}{\nu _{n}}\int_{|t|\leq \nu _{n}}d\tau \int_{-\nu _{n}-\tau
}^{\nu _{n}-\tau }v^{2}(y)dy\int_{|\theta |\geq \nu _{n}}v^{2}(\theta -\tau
)d\theta \\
& \leq \frac{2w^{2}}{\nu _{n}}\int_{|\tau |\leq \nu _{n}}d\tau
\int_{|\theta |\geq \nu _{n}}v^{2}(\theta -\tau )d\theta .
\end{align*}%
This and (\ref{T3}) yield%
\begin{equation}
\lim_{n\rightarrow \infty }\frac{1}{2\nu _{n}}\int_{|t|\leq \nu
_{n}}T_{2}^{(n)}dt=0.  \label{T2}
\end{equation}%
We also have
\begin{align}
\frac{1}{2\nu _{n}}& \int_{|t|\leq \nu _{n}}T_{1}^{(n)}dt\leq \frac{2}{\nu
_{n}}\int_{|t|\leq \nu _{n}}dt\int_{|\tau |\leq \nu _{n}}v^{2}(t-\tau )d\tau
\notag \\
& \hspace{5cm}\times \int_{|\theta |\leq \nu _{n}}v^{2}(\theta -\tau
)|f_{n}(\theta ,z)-h(z)|d\theta  \notag \\
& =\frac{2}{\nu _{n}}\int_{|\theta |\leq \nu _{n}}|f_{n}(\theta
,z)-f(z)|d\theta \int_{|\tau |\leq \nu _{n}}v^{2}(\theta -\tau )d\tau
\int_{|t|\leq \nu _{n}}v^{2}(t-\tau )d\tau  \notag \\
& \leq \frac{2w^{4}}{\nu _{n}}\int_{|\theta |\leq \nu _{n}}|f_{n}(\theta
,z)-f(z)|d\theta .  \label{T1}
\end{align}%
It follows from (\ref{T123}) and (\ref{T3}) -- (\ref{T1}) that
\begin{equation*}
(1-4w^{4}K_{0}^{-2})\frac{1}{2\nu _{n}}\int_{|t|\leq \nu
_{n}}|f_{n}(t,z)-f(z)|dt=o(1),\quad n\rightarrow \infty .
\end{equation*}%
We conclude that if $\nu =\infty $, then $f$ of (\ref{gnjf}) coincides with
the solution of (\ref{eqc}) satisfying (\ref{Imzf}). Thus to finish the
proof of Theorem \ref{t:main} it remains to prove

\begin{lemma}
\label{l:aux}Denote $\mathbf{E}_{k}$ the expectation with respect to $%
\mathbf{y}_{k}$ and let for any random variable $\xi $
\begin{equation*}
\xi _{k}^{\circ }=\xi -\mathbf{E}_{k}\{\xi \}
\end{equation*}%
be its centered version, then we have under the conditions of Theorem \ref%
{t:main}:

(i)
\begin{equation}
\;\mathbf{E}\{|(G^{k}\mathbf{y}_{k},\mathbf{y}_{k})_{k}^{\circ }|^{2}\}\leq
C(z)\delta _{n},  \label{Gk0k}
\end{equation}

(ii) there exists $K_{0}>0$ such that $\forall z\in C_{K_{0}}=\{z:\;\Re
z=0,\;|\Im z|\geq K_{0}\}$
\begin{align}
& \mathbf{Var}\{G_{pp}\}\leq C(z)\delta _{n},  \label{Vpp} \\
& \mathbf{Var}\{(G^{k}\mathbf{y}_{k},\mathbf{y}_{k})\}\leq C(z)\delta _{n}.
\label{Gyy}
\end{align}%
where
\begin{equation}
\delta _{n}=o(1),\quad n\rightarrow \infty  \label{den}
\end{equation}%
does not depend on $p$, $k$, $z$, and we denote by $C(z)$ any positive
quantity, which depends only on $z$ and is finite for $z\in C_{K_{0}}$ (see (%
\ref{CK0}))
\end{lemma}

\begin{proof}
It follows from (\ref{m22}) and from unconditionality of the distribution of
$\mathbf{a}_{k}$ that
\begin{equation}
\mathbf{E}\{a_{pk}^{(n)}a_{qk}^{(n)}a_{sk}^{(n)}a_{tk}^{(n)}%
\}=m_{2,2}^{(n)}(\delta _{pq}\delta _{st}+\delta _{ps}\delta _{qt}+\delta
_{pt}\delta _{qs})+\kappa _{4}^{(n)}\delta _{pq}\delta _{ps}\delta _{pt},
\label{aaaa}
\end{equation}%
where $\kappa _{4}^{(n)}=m_{4}^{(n)}-3(m_{2,2}^{(n)})^{2}=O(1)$, $%
n\rightarrow \infty $. This and (\ref{EABk}) yield
\begin{align}
& \mathbf{E}_{k}\{|(G^{k}\mathbf{y}_{k},\mathbf{y}_{k})_{k}^{\circ }|^{2}\}=
\label{Gyy=} \\
& =\frac{1}{b_{n}^{2}}\sum_{p,q,s,t}G_{pq}\overline{G_{st}}%
v_{pk}v_{qk}v_{sk}v_{tk}\mathbf{E}_{k}%
\{a_{pk}^{(n)}a_{qk}^{(n)}a_{sk}^{(n)}a_{tk}^{(n)}\}-\Big|\frac{1}{b_{n}}%
\sum_{p}v_{pk}^{2}G_{pp}^{k}\Big|^{2}  \notag \\
& =(m_{2,2}^{(n)}-1)\Big|\frac{1}{b_{n}}\sum_{p}v_{pk}^{2}G_{pp}^{k}\Big|%
^{2}+\frac{2m_{2,2}^{(n)}}{b_{n}^{2}}%
\sum_{p,s}v_{sk}^{2}v_{pk}^{2}|G_{ps}^{k}|^{2}+\frac{\kappa _{4}^{(n)}}{%
b_{n}^{2}}\sum_{p}v_{pk}^{4}|G_{pp}^{k}|^{4}.  \notag
\end{align}%
By (\ref{G<}) and (\ref{vK}) we have
\begin{align*}
& \Big|\frac{1}{b_{n}}\sum_{p}v_{pk}^{2}G_{pp}^{k}\Big|\leq \frac{K}{|\Im z|}%
, \\
& \frac{1}{b_{n}^{2}}\sum_{p,s}v_{sk}^{2}v_{pk}^{2}|G_{ps}^{k}|^{2}\leq
\frac{K}{b_{n}^{2}}\sum_{p}v_{pk}^{2}\sum_{s}|G_{ps}^{k}|^{2}\leq \frac{K^{2}%
}{b_{n}|\Im z|^{2}}.
\end{align*}%
This, (\ref{m22}) and (\ref{Gyy=})
lead to (\ref{Gk0k}).

Let us prove (\ref{Vpp}). We have (cf. (\ref{EGjj}))
\begin{align}
\mathbf{E}\{G_{jj}\overline{G_{jj}^{\circ }}\}=\frac{1}{z}\sum_{k}\mathbf{E}%
\Big\{\frac{B_{kn}}{A_{kn}}\overline{G_{jj}^{\circ }}\Big\}=& \frac{1}{z}%
\sum_{k}\mathbf{E}\Big\{\mathbf{E}_{k}\Big\{\frac{B_{kn}}{A_{kn}}\Big\}%
\overline{G_{jj}^{k\circ }}\Big\}  \label{T12} \\
& +\frac{1}{z}\sum_{k}\mathbf{E}\Big\{\frac{B_{kn}}{A_{kn}}\overline{%
(G_{jj}-G_{jj}^{k})^{\circ }}\Big\}=:\mathcal{T}_{1}+\mathcal{T}_{2}.  \notag
\end{align}%
It follows from the Schwarz inequality, (\ref{EA>}), (\ref{Bk<}) and (\ref%
{GGk}):
\begin{equation}
|\mathcal{T}_{2}|\leq \sum_{k}\mathbf{E}\{|B_{k}|^{2}\}^{1/2}\;\mathbf{E}%
\{|(G_{pp}-G_{pp}^{k})^{\circ }|^{2}\}^{1/2}\leq C(z)b_{n}^{-1/2}.
\label{T2c}
\end{equation}%
Consider now $\mathcal{T}_{1}$ of (\ref{T2c}). We have by (\ref{A})
\begin{equation*}
\mathbf{E}_{k}\Big\{\frac{B_{kn}}{A_{kn}}\Big\}=\frac{\mathbf{E}%
_{k}\{B_{kn}\}}{\mathbf{E}_{k}\{A_{kn}\}}-\frac{1}{\mathbf{E}_{k}\{A_{kn}\}}%
\mathbf{E}_{k}\Big\{\frac{(A_{kn})_{k}^{\circ }B_{kn}}{A_{kn}}\Big\},
\end{equation*}%
and by the Schwarz inequality, (\ref{EA>}), (\ref{Bk<}), and (\ref{Gk0k})
\begin{equation*}
\mathbf{E}_{k}\Big\{{{A_{kn}}^{-1}(A_{kn})_{k}^{\circ }B_{kn}}\Big\}\leq
C(z)b_{n}^{-1}|v_{jk}|\delta _{n}^{1/2}.
\end{equation*}%
This and (\ref{EABk}) yield
\begin{equation}
\mathcal{T}_{1}=\frac{1}{zb_{n}}\sum_{k}v_{jk}^{2}\mathbf{E}\Big\{\frac{1}{\mathbf{E}%
_{k}\{A_{kn}\}}G_{jj}^{k}\overline{G_{jj}^{k\circ }}\Big\}+o(1),\quad
n\rightarrow \infty .  \label{T1=}
\end{equation}%
Applying again (\ref{A}) and then (\ref{EABk}), we get
\begin{align*}
\mathbf{E}\Bigg\{& \frac{G_{jj}^{k}\overline{G_{jj}^{k\circ }}}{\mathbf{E}%
_{k}\{A_{kn}\}}\Bigg\}=\frac{1}{\mathbf{E}\{A_{kn}\}}\mathbf{E}\Big\{%
|G_{jj}^{k\circ }|^{2}\Big\}-\frac{1}{\mathbf{E}\{A_{kn}\}}\mathbf{E}\Big\{%
\frac{(\mathbf{E}_{k}\{A_{kn}\})^{\circ }}{\mathbf{E}_{k}\{A_{kn}\}}%
G_{jj}^{k}\overline{G_{jj}^{k\circ }}\Big\} \\
& \quad \quad =\frac{1}{\mathbf{E}\{A_{kn}\}}\mathbf{E}\Big\{|G_{jj}^{k\circ
}|^{2}\Big\}-\frac{1}{\mathbf{E}\{A_{kn}\}}\frac{1}{b_{n}}\sum_{p}v_{pk}^{2}%
\mathbf{E}\Big\{\frac{G_{pp}^{k\circ }}{\mathbf{E}_{k}\{A_{kn}\}}G_{jj}^{k}%
\overline{G_{jj}^{k\circ }}\Big\}.
\end{align*}%
Note also that in view of (\ref{GGk}) we can replace $G^{k}$ with $G$ with
the error term of the order $O(b_{n}^{-1})$, hence
\begin{align*}
\mathcal{T}_{1}=& \frac{1}{zb_{n}}\sum_{k}v_{jk}^{2}\frac{1}{\mathbf{E}\{A_{kn}\}}%
\cdot \mathbf{E}\Big\{|G_{jj}^{\circ }|^{2}\Big\} \\
& -\frac{1}{zb_{n}^{2}}\sum_{k}v_{jk}^{2}\frac{1}{\mathbf{E}\{A_{kn}\}}%
\sum_{p}v_{pk}^{2}\mathbf{E}\Big\{\frac{G_{jj}^{k}}{\mathbf{E}_{k}\{A_{kn}\}}%
G_{pp}^{\circ }\overline{G_{jj}^{\circ }}\Big\}+o(1),\quad n\rightarrow
\infty ,
\end{align*}%
and by the Schwarz inequality, (\ref{EA>}), and (\ref{vK})
\begin{equation*}
|\mathcal{T}_{1}|\leq \frac{K}{|\Im z|}\mathbf{Var}\{G_{jj}\}+\frac{K^{2}}{|\Im z|^{2}}%
\mathbf{Var}\{G_{jj}\}^{1/2}\max_{|p|\leq m}\mathbf{Var}\{G_{pp}%
\}^{1/2}+o(1),\quad n\rightarrow \infty .
\end{equation*}%
This and (\ref{T12}) -- (\ref{T2c}) yield for $V_{j}:=\mathbf{Var}\{G_{jj}\}$,
$|j|\leq m,\;z\in C_{K_{0}}$:
\begin{equation*}
V_{j}\leq \frac{K}{K_{0}}V_{j}+\frac{K^{2}}{K_{0}^{2}}V_{j}^{1/2}\max_{|p|%
\leq m}V_{p}^{1/2}+o(1),\quad n\rightarrow \infty .
\end{equation*}%
Choosing here $K_{0}$ such that%
\begin{equation}
K/K_{0}+K^{2}/K_{0}^{2}<1,  \label{K2}
\end{equation}
we obtain that%
\begin{equation*}
\max_{|p|\leq m}V_{p}=o(1),\;n\rightarrow \infty ,
\end{equation*}
i.e., (\ref{Vpp}).

It remains to note that we have by (\ref{EABk})
\begin{equation*}
(G^{k}\mathbf{y}_{k},\mathbf{y}_{k})^{\circ }=(G^{k}\mathbf{y}_{k},\mathbf{y}%
_{k})_{k}^{\circ }+b_{n}^{-1}\sum_{p}v_{pk}^{2}G_{pp}^{k\circ }.
\end{equation*}%
This together with (\ref{Gk0k}) -- (\ref{Vpp}) lead to (\ref{Gyy}) and
complete the proof of the lemma.
\end{proof}

\section{Proof of Theorem \protect\ref{t:indep}}

\label{s:indep}

The proof of Theorem \ref{t:indep} can be obtained by following the scheme
worked out in \cite{Ly-Pa:08} (see also \cite{Pa-Sh:11} and references
therein) and applicable to a wide variety of random matrices with
independent entries. Namely, one uses first the martingale-type argument to
prove the bound (\ref{EN}) and then the so-called interpolation trick to
reduce the initial problem to that one of finding the limit (\ref{gnjf}) for
the random matrices with Gaussian entries with the same first and second
moments. Since these two steps are rather standard, we will explain below
just the derivation of the limiting equations (\ref{ftz}) -- (\ref{Imzf})\
for i.i.d. Gaussian entries $\{a_{jk}\}_{|j|,|k|\leq m}$ satisfying (\ref%
{iso}). Note also that by using a standard truncation technique condition (%
\ref{a2e}) can be replaced with the Lindeberg type condition for the second
moments (see \cite{Ly-Pa:08,Pa-Sh:11}).

Accordingly, consider a random matrix $A_{n}$ % \begin{equation}
% A_{n}=\{A^{(n)}_{jk }\}_{ |j|,|k|\leq m}, \quad A^{(n)}_{jk }=b_n^{-1/2}v\big({(j-k)}/{b_{n}}\big)a^{}_{jk},
% %\quad  |j|,|k|\leq m.
% \label{Ajk}
% \end{equation}
(\ref{Ajk}) -- (\ref{K}), where $\{a_{jk}\}_{|j|,|k|\leq m}$ are jointly
independent standard Gaussian random variables of zero mean and unit
variance. We will use

\begin{proposition}
\label{p:Nash} Let $\xi=\{\xi _{l}\}_{l=1}^{p}$ be independent Gaussian
random variables of zero mean, and $\Phi :\mathbb{R}^{p}\rightarrow \mathbb{C%
}$ be a differentiable function with polynomially bounded partial
derivatives $\Phi _{l}^{\prime },\;l=1,...,p$. Then we have
\begin{equation}
\mathbf{E}{\mathbb{\{\xi }}_{l}{\Phi (\xi )}\}=\mathbf{\ E}{\mathbb{\{}{\xi }%
}_{l}^{2}\}\mathbf{E}\mathbb{\{}{\Phi _{l}^{\prime }(\xi )}\},\;l=1,...,p,
\label{diffga}
\end{equation}%
and%
\begin{equation}
\mathbf{Var}\{\Phi (\xi) \}\leq \sum_{l=1}^{p}\mathbf{E}\{\xi _{l}^{2}\}%
\mathbf{E}\left\{ |\Phi _{l}^{\prime }(\xi)|^{2}\right\} .  \label{Nash}
\end{equation}
\end{proposition}

The first formula is a version of the integration by parts. The second is a
version of the Poincar\'{e} inequality (see e.g. \cite{Pa-Sh:11}).

We have by the resolvent identity and (\ref{diffga})
\begin{equation*}
\mathbf{E}\{G_{jj}(z)\}=-\frac{1}{z}+\frac{1}{zb_{n}}\sum_{k}v_{jk}^{2}%
\mathbf{E}\{D_{jk}(G(z)A_{n})_{jk}\},
\end{equation*}%
where $D_{jk}=\partial /\partial A_{jk}^{(n)}$. It can be shown that
\begin{equation}
D_{jk}(GA_{n})_{jk}=-z\widetilde{G}_{kk}G_{jj}-(GA_{n})_{jk}^{2}  \label{DGA}
\end{equation}%
(see, e.g., \cite{Ly-Pa:08}), where $\widetilde{G}=(\widetilde{M}%
-zI_{n})^{-1}$, $\widetilde{M}=A_{n}^{T}A_{n}$. Hence,
\begin{align}
\mathbf{E}\{G_{jj}(z)\}& =-\frac{1}{z}-\frac{1}{b_{n}}\sum_{k}v_{jk}^{2}%
\mathbf{E}\{\widetilde{G}_{\alpha \alpha }(z)G_{jj}(z)\}+r_{n}(z),
\label{GG} \\
r_{n}(z)& =-\frac{1}{b_{n}}\sum_{k}v_{jk}^{2}\mathbf{E}\{(GA_{n})_{jk}^{2}\}.
\notag
\end{align}%
Since
\begin{equation*}
\sum_{k}\big|(GA_{n})_{jk}^{2}\big|=(GM_{n}\overline{G})_{jj}=((I_{n}+zG)%
\overline{G})_{jj},
\end{equation*}%
then
\begin{equation}
r_{n}=O(b_{n}^{-1}),\quad n\rightarrow \infty .  \label{rn1}
\end{equation}%
We also have by (\ref{Nash}) -- (\ref{DGA})
\begin{align}
\mathbf{Var}\{G_{jj}(z)\}& \leq \sum_{l,k}\mathbf{E}\{|\partial
G_{jj}/\partial a_{lk}|^{2}\}  \label{VG} \\
& \leq \frac{4}{b_{n}}\sum_{l,k}v_{lk}^{2}\mathbf{E}%
\{|(GA_{n})_{jl}G_{jk}|^{2}\}\leq C(z)b_{n}^{-1}.  \notag
\end{align}%
Now it follows from (\ref{GG}) -- (\ref{VG}) that
\begin{equation}
\mathbf{E}\{G_{jj}(z)\}=-\frac{1}{z}-\frac{1}{b_{n}}\sum_{k}v_{jk}^{2}%
\mathbf{E}\{\widetilde{G}_{kk}(z)\}\mathbf{E}\{G_{jj}(z)\}+O(b_{n}^{-1/2}),
\label{1}
\end{equation}%
as $n\rightarrow \infty $. Similarly,
\begin{equation}
\mathbf{E}\{\widetilde{G}_{kk}(z)\}=-\frac{1}{z}-\frac{1}{b_{n}}%
\sum_{p}v_{pk}^{2}\mathbf{E}\{G_{pp}(z)\}\mathbf{E}\{\widetilde{G}%
_{kk}(z)\}+O(b_{n}^{-1/2}),  \label{2}
\end{equation}%
as $n\rightarrow \infty $. Solving system (\ref{1}) -- (\ref{2}), we get (%
\ref{EGjj=}) and then it suffices to use the same argument as that in the
proof of Theorem \ref{t:main} to finish the proof of Theorem \ref{t:indep}.

\section{Triangular matrices}

\label{triang}

In this section we prove Theorem \ref{t:tri}. It follows from (\ref{eq})
that if $v(t)=\chi _{\lbrack 0;1]}(t)$ and $\nu =1/2$, then
\begin{equation}
f(t,z)=-\Big(z-\int_{-1/2}^{t}\big(1+\int_{\tau }^{1/2}f(\theta ,z)d\theta %
\big)^{-1}d\tau \Big)^{-1}.  \label{fttz}
\end{equation}%
Denote
\begin{equation}
\varphi (t,z)=\int_{t}^{1/2}f(\theta ,z)d\theta .  \label{phi}
\end{equation}%
It follows from (\ref{fttz}) -- (\ref{phi}) that%
\begin{equation*}
\varphi ^{\prime \prime }-\varphi ^{\prime 2}(1+\varphi )^{-1}=0,\;\varphi
(1/2,z)=0,\;\varphi ^{\prime }(1/2,z)=z^{-1},
\end{equation*}
where $\varphi ^{\prime }=\partial \varphi /\partial t$. Solving this
system, we get for $f(z)=\varphi (-1/2,z)$: % \begin{equation*}
% \begin{cases}
% \varphi(t,z) =e^{-(t-1)c(z)}-1,&  \\
% c(z)e^{c(z)}=-z^{-1}. &  \\
% \end{cases}
% \end{equation*}
% Plugging $f(t,z)=-\varphi'(t,z)=c(z)e^{-(t-1)c(z)}=-z^{-1}e^{-tc(z)} $ %in (\ref{fftz}), we finally get
\begin{equation}
f(z)=e^{c(z)}-1,\;c(z)e^{c(z)}=-z^{-1},\;z\in \mathbb{C}\backslash \lbrack
0;\infty ).  \label{sys}
\end{equation}%
These equations are equivalent to (\ref{eqtri}). Evidently, there is only
one solution $c$ analytical in $\mathbb{R}\setminus \lbrack 0,\infty )$.

Let us prove (\ref{supp}). As it was firstly shown in \cite{Ma-Pa:67} (see
also \cite{Ba-Si:10,Pa-Sh:11}), to find the support of measure $N$, it
suffices to consider function $x=x(f),f\in \mathbb{R}$, which is the
functional inverse of Stieltjes transform of $N$, and to find set $L\subset
\mathbb{R}$ on which $x$ increases monotonically. Then $\mathrm{supp}\;N=%
\overline{\mathbb{R}\setminus x(L)}$, where $x(L)=\{x(f):\;f\in L\}$.

It follows from (\ref{sys}) that in our case
\begin{equation}
x(f)=-\frac{1}{(1+f)\ln (1+f)},\quad f>-1,\;f\neq 0.  \label{xf}
\end{equation}%
It is easy to find that $x(f)$ increases on $L=[e^{-1},0)\cup (0,\infty )$.
Thus $x(L)=(-\infty ,0)\cup \lbrack e,\infty )$ and $\mathrm{supp}\;N=%
\overline{\mathbb{R}\setminus x(L)}=[0,e]$.

To prove asymptotic relations (\ref{N0}), we first consider
\begin{equation*}
F(x):=f(-x)=\int_{0}^{\infty }\frac{N(d\lambda )}{\lambda +x},\quad x>0.
\end{equation*}%
It is easy to find from (\ref{eqtri}) that
\begin{equation}
F(x)=\frac{1}{x\ln 1/x}(1+o(1)),\quad x\downarrow 0.  \label{F0}
\end{equation}%
This and the Tauberain theorem (see \cite{F:08}, Chapter XIII.5) imply
\begin{equation*}
N(\lambda ):=N([0,\lambda ])=\frac{1}{\ln 1/\lambda }(1+o(1)),\quad
\lambda \downarrow 0.
\end{equation*}%
Differentiating formally this asymptotic formula we obtain the first formula of (%
\ref{N0}).%
%\begin{lemma}\label{l:varjj}
%Under conditions of  Theorem \ref{t:main},  bound (\ref{Varjj}) is valid. \end{lemma}
%\begin{proof}
%Applying the Poincar\'e-Nash
%inequality (see e.g. \cite{Pa-Sh:11}, Chapter 2.1),  formula
% \begin{align*}
%D_{k\alpha }G_{jj }=-2(GA_n)_{j\alpha}G_{jk},\quad D_{k\alpha }=\partial/\partial A_{k\alpha },
%%\label{DGA}
%\end{align*}
%(\ref{sum}), and bound $||G(z)||\leq |\Im z|^{-1}$, we get
%\begin{align*}
%\mathbf{Var}\{G_{jj}(z)\}&\le\sum_{1\leq\alpha<k\leq n}  \mathbf{E}\{|\partial G_{jj }/\partial %a_{k\alpha }|^{2}\} %\label{N}
% \\
% &\le\frac{1}{n}\sum_{\alpha, k=1}^n  \mathbf{E}\{|D_{k\alpha }G_{jj }|^{2}\}=\frac{4}%{n}\sum_{\alpha, k=1}^n  \mathbf{E}\{|(GA_n)_{j\alpha}G_{jk}|^{2}\}\leq C(z)n^{-1}.
%\end{align*}
%The technique proposed in \cite{Sh:10} allows to get  (\ref{Varjj}) in general
%case of non-Gaussian entries. Thus the proof of the lemma is complete, and so does the proof of %Theorem \ref{t:main}.

To prove this formula rigorously we use (\ref{sys}). Denoting $c(\lambda
+i0)=\xi (\lambda )+i\eta (\lambda )$, we obtain from (\ref{sys}) and (\ref%
{SP})

\begin{equation}
\rho (\lambda )=\frac{\sin ^{2}\eta }{\pi \lambda \eta },\quad \lambda =%
\frac{e^{\eta \cot \eta }\sin \eta }{\eta },\quad \eta \in \lbrack 0,\pi ].
\label{rol}
\end{equation}%
Since the limit $\lambda \downarrow 0$ corresponds to $\eta =\pi -\sigma
,\;\sigma \downarrow 0$, we have from (\ref{rol})  $\ln\lambda =-\pi /\sigma
+O(\ln \sigma),\;\sigma \downarrow 0$ and eventually the first asymptotics of (\ref%
{N0}). The second asymptotics of (\ref{N0}) can be obtained similarly taking
into account that the limit $\lambda \uparrow e$ corresponds to the limit $%
\eta \downarrow 0$.

Let us prove (\ref{moms}). To this end we will use the identity%
\begin{equation*}
\mu _{k}:=\frac{k^{k}}{(k+1)!}=\frac{1}{2\pi ik}\int_{|\zeta |=\varepsilon }%
\frac{e^{k\zeta }}{\zeta ^{k+2}}d\zeta ,\quad \varepsilon >0, \quad k\geq 1.
\end{equation*}%
Consider the generating function%
\begin{equation*}
h(z)=-\sum_{l=0}^{\infty }\frac{\mu _{l}}{z^{l+1}},
\end{equation*}%
which is well defined if $z$ is sufficiently large. We have then the
integral representation%
\begin{equation*}
zh(z)=-1+\frac{1}{2\pi i}\int_{|\zeta |=\varepsilon }\log \left( 1-\frac{%
e^{\zeta }}{\zeta z}\right) \frac{d\zeta }{\zeta ^{2}}
\end{equation*}%
or%
\begin{equation}
(zh(z))^{^{\prime }}=\frac{1}{2\pi iz}\int_{|\zeta |=\varepsilon }\frac{%
d\zeta }{\zeta ^{2}(e^{-\zeta }\zeta z-1)}.  \label{fzp}
\end{equation}%
The both formulas are valid for $|z|>e^{\varepsilon }\varepsilon ^{-1}$,
where the integrands are analytic in $\zeta $ just because the series for
the integrands are convergent. It is easy to find that the function $%
u_{z}(\zeta )=e^{-\zeta }\zeta z-1$ has a simple zero in $\zeta
(z)=z^{-1}(1+o(1)),\;z\rightarrow \infty $, i.e., inside the contour $|\zeta
|=\varepsilon$ if $\varepsilon $ does not depend on $z$. Thus the integral
on the left of (\ref{fzp}) is equal to $z^{-1}$ times the residue of $u_{z}$
at $\zeta (z)$ (i.e., $(\zeta (z)(1-\zeta (z)))^{-1}$) plus the integral
over a sufficiently "small" contour, say $|\zeta |=|2z|^{-1}$. Since $u_{z}$
has no zeros inside this contour, the corresponding integral is just $%
u_{z}^{\prime }(0)=-z$. Putting everything together, we obtain%
\begin{equation*}
(zh(z))^{\prime }=(z\zeta (z)(1-\zeta (z)))^{-1}-1,\;z\zeta (z)e^{-\zeta
(z)}=1.
\end{equation*}%
On the other hand, it follows from (\ref{sys}) that $(zf(z))^{\prime
}=-(zc(z)(1+c(z)))^{-1}-1$. Thus, setting $\zeta =-c$ we obtain that $h$
coincides with $f$, i.e., assertion (iv) of the theorem.

\end{document}